\documentclass[conference]{IEEEtran}

\IEEEoverridecommandlockouts

\usepackage[utf8]{inputenc} 
\usepackage[T1]{fontenc}
\usepackage{url}
\usepackage{cite}
\usepackage[cmex10]{amsmath} 

\usepackage{amssymb}
\usepackage{xcolor}
\usepackage{graphicx}
\usepackage{dirtytalk}

\renewcommand{\le}{\leqslant}
\renewcommand{\leq}{\leqslant}
\renewcommand{\ge}{\geqslant}
\renewcommand{\geq}{\geqslant}

\newcommand{\Cref}[1]{Co\-ro\-lla\-ry\,\ref{#1}}



\outer\def\proclaim #1. #2\par{\medbreak
 \noindent{\bf#1.\enspace}{\sl#2\par}%
 \ifdim\lastskip<\medskipamount \removelastskip\penalty55\medskip\fi}

\newtheorem{theorem}{Theorem}

\newtheorem{definition}{Definition}

\newtheorem{lemma}{Lemma}
\newtheorem{corollary}{Corollary}

\newtheorem{conjecture}{Conjecture}

\makeatletter
\def\@begintheorem#1#2{%
  \trivlist
  \item[\hskip\labelsep{\bfseries #1\ #2.}]%
}
\def\@opargbegintheorem#1#2#3{%
  \trivlist
  \item[\hskip\labelsep{\bfseries #1\ #2\ (#3).}]%
}
\makeatother

\begin{document}

\title{A Hyperbolic Bound for File Retrieval\\ in DNA-Based Data Storage%
\thanks{The work of D. Bar-Lev was supported in part by the
Swiss National Science Foundation under Grant 212865 and by
Schmidt Sciences.}}

\author{%
\IEEEauthorblockN{%
Constantinos Vasilios Argyris Vlachos\IEEEauthorrefmark{1}\IEEEauthorrefmark{3}
and Daniella Bar-Lev\IEEEauthorrefmark{1}\IEEEauthorrefmark{2}}
\IEEEauthorblockA{%
\IEEEauthorrefmark{1}\textit{Department of Mathematics,
Universit\"at Z\"urich}, Z\"urich, Switzerland\\
\IEEEauthorrefmark{3}\textit{IBM Research Europe}, Z\"urich, Switzerland\\
\IEEEauthorrefmark{2}\textit{Department of Computer Science,
Technion---Israel Institute of Technology}, Haifa 3200003, Israel\\
constantinosvasiliosargyris.vlachos@uzh.ch,\quad daniellalev@technion.ac.il}}

\maketitle

\begin{abstract}
In DNA-based storage systems, data are retrieved by sequencing molecules sampled from a DNA pool. We study how the way coding redundancy is shared between files affects their expected retrieval times. Our focus lies on the case of two files that are encoded by a systematic linear code over an arbitrary finite field. We consider the conjecture that the sum obtained by dividing each file dimension by its expected retrieval time is at most one whenever the dimension of at least one file is more than one. For this, we develop a geometric view of the retrieval process. As molecules are sampled, we follow the growing span of the corresponding columns of the generator matrix and track how much of this span comes from each file. Among the samples that enlarge the overall span, this lets us compare those that make progress toward recovering both files with those that make progress toward neither. We call a column mixed if the corresponding encoded symbol combines information from both files. In this work, we sharpen a projection bound and use it to control the effect of
mixed columns. We prove the conjecture whenever the total information dimension is at least twice the number of mixed columns plus two. For equal-sized files, this allows up to one fewer mixed column than the dimension of either file and extends the previous result for codes with no mixed columns.
\end{abstract}


\section{Introduction}
DNA-based data storage has emerged as a promising medium for long-term archival storage due to its high information density and durability. The feasibility of DNA-based storage solutions has been demonstrated in experiments ranging from early proof-of-concept systems to large-scale and recent end-to-end implementations; see, e.g., \cite{church2012next,goldman2013towards,grass2015robust,erlich2017dna,organick2018random,bar2025scalable}. In a typical system, digital information is encoded into a collection of short DNA sequences, which are synthesized as molecules and stored together in an unordered storage container. To read the archive, molecules from the pool are sequenced, and the resulting reads are processed to reconstruct the designed sequences and decode the stored information. For further background on coding techniques and problems in DNA-based data storage, we refer the reader to the recent surveys~\cite{milenkovic2024dna,sabary2024survey}.

In DNA-based data storage, each encoded sequence is represented in the pool by many molecular copies. The sequencing process is therefore modeled by drawing encoded strands independently with replacement according to a probability distribution determined by the storage channel~\cite{heckel2019characterization}. We consider the uniform noiseless model, in which every encoded strand is sampled with the same probability and is observed without error. The resulting retrieval process is closely related to the classical Coupon Collector Problem~\cite{erdHos1961classical,flajolet1992birthday}. The presence of coding changes the stopping condition, since recovery occurs once the sampled encoded sequences contain enough information to decode the requested data, even if some sequences have not been observed. This coded sampling problem was introduced in~\cite{barlev2025cover}, which considered two extreme recovery objectives. 
In the first, the entire information vector must be recovered; this full-recovery problem was further studied in, e.g., ~\cite{hanna2026reliability,cao2026modeling,bertuzzo2026duality,grunbaum2026general}. In the second, only one information symbol is requested; this random-access problem was explored, for example,  in~\cite{gruica2025combinatorial,gruica2026geometry,boruchovsky2026making,bodur2026random,wang2026random}.

In practice, stored data are organized into files, and a user may wish to recover one or several files rather than a single information symbol or the entire database. Abraham et al.~partitioned the information into equal-sized files and studied the maximum expected time required to retrieve any requested collection of a fixed number of files~\cite{abraham2024covering}. More recently, Gruica et al.~allowed any number of information symbols to be requested and studied the maximum and average expected retrieval times over all requested sets of the same size, without fixing a partition into files~\cite{gruica2026generalized}. The worst-case objective used in these formulations is a natural criterion for balanced performance, since a code that improves the retrieval time for one request may perform worse for another. However, it does not describe the tradeoff between the retrieval times of specific files. In~\cite{barlev26coded}, the information was partitioned into two fixed files of arbitrary dimensions, and the objective was instead to characterize the achievable pairs of their expected retrieval times. This is the model considered here.

Let $E_i(G)$ denote the expected number of samples required to recover file $F_i$, whose dimension is $s_i$. It was conjectured in~\cite{barlev26coded} that every linear code satisfies
\[
\frac{s_1}{E_1(G)}+\frac{s_2}{E_2(G)}\leq 1
\]
whenever at least one file has dimension at least two. We refer to this inequality as the hyperbolic bound. The bound is known when every column of $G$ lies in one of the two file subspaces. It is also asymptotically tight, since file-dedicated MDS codes approach equality as the length grows. The main difficulty comes from columns with nonzero components in both file subspaces, which we call \emph{mixed columns}. Such a column may contribute to the recovery of both files, so an argument that treats the two retrieval processes separately may count its contribution twice.

In this work, we prove the hyperbolic bound for systematic linear codes over an arbitrary finite field, with no zero columns, whenever $k=s_1+s_2\geq 2M+2$, where $M$ is the number of mixed columns. We first introduce a defect that tracks the interaction between the two file subspaces and compare samples that advance both files with those that advance neither. Then we sharpen a projection bound and use it to control the contribution of the mixed columns.
For two files of equal dimension $s$, our result applies whenever $M\leq s-1$.

The remainder of the paper is organized as follows. Section~\ref{sec:prelim} introduces the notation and problem formulation. Section~\ref{sec:geometric} develops the geometric argument, and Section~\ref{sec:mixed_bounds} proves the bound in terms of the number of mixed columns. Section~\ref{sec:conclusion} concludes the paper.
 
\section{Notation and Problem Formulation}\label{sec:prelim}
\subsection{Notations}

Throughout, $k$ and $n$ are positive integers with $k \leq n$, $q$ is a prime power, $\mathbb{F}_q$ is the finite field with $q$ elements, and $[m] := \{1, \ldots, m\}$. The $r$-th harmonic number is $H_r := \sum_{i=1}^r \frac{1}{i}$, with $H_0 = 0$.

A generator matrix $G \in \mathbb{F}_q^{k \times n}$ of rank $k$ encodes $k$ information symbols into $n$ codeword symbols. Let $g_j$ denote the $j$-th column of $G$.
Throughout this work, we restrict attention to systematic codes and, after reordering the columns if necessary, write $G=[\,I_k\mid P\,]$.

Let $e_1,\ldots,e_k$ denote the standard basis vectors of $\mathbb F_q^k$. The information space $\mathbb{F}_q^k$ is partitioned into two files $F_1$ and $F_2$, where
$F_1 = \langle e_1, \ldots, e_{s_1} \rangle$ and $F_2 = \langle e_{s_1+1}, \ldots, e_k \rangle$,
and $s_1,s_2$ are positive integers satisfying $s_1+s_2=k$. In particular, $F_1\oplus F_2=\mathbb{F}_q^k$.

\begin{definition}
A subset $S \subseteq [n]$ forms a \emph{recovery set} for file $F_i$ if $F_i \subseteq \langle g_j : j \in S \rangle$.
\end{definition}

During retrieval, column indices
$\xi_1,\xi_2,\ldots$ are drawn independently and uniformly from $[n]$. For $i\in\{1,2\}$, the retrieval time of $F_i$ is
$\tau_{F_i}(G) :=
\min\left\{ t\geq 0: F_i\subseteq
\left\langle g_{\xi_1},\ldots,g_{\xi_t}
\right\rangle
\right\}$,
and its expected retrieval time is denoted by
\[
E_i(G):=\mathbb{E}[\tau_{F_i}(G)].
\]

Since $F_1 \oplus F_2 = \mathbb{F}_q^k$, every vector $v \in \mathbb{F}_q^k$ decomposes uniquely as $v = \pi_1(v) + \pi_2(v)$, where $\pi_i : \mathbb{F}_q^k \to F_i$ is the projection onto $F_i$ along $F_{3-i}$. Concretely, $\pi_1$ retains the first $s_1$ coordinates and zeros the rest; $\pi_2$ retains the last $s_2$ coordinates. We assume throughout that $G$ has no zero columns. A column $g_j$ is \emph{pure-$F_i$} if $\pi_{3-i}(g_j) = 0$ (equivalently $g_j \in F_i$), and \emph{mixed} if both projections are nonzero.

\begin{definition}
For any subspace $W \subseteq \mathbb{F}_q^k$,
\[
N(W) := |\{j \in [n] : g_j \in W\}|.
\]
\end{definition}

\subsection{The Universal Hyperbolic Bound}

The following conjecture was posed in~\cite{barlev26coded} for arbitrary
linear codes.

\begin{conjecture}[Universal Hyperbolic Bound]\label{conj:hyperbolic}
For any partition $s_1+s_2=k$ with $\max\{s_1,s_2\}\geq2$ and any
rank-$k$ generator matrix $G\in\mathbb F_q^{k\times n}$,
\begin{equation}
\frac{s_1}{E_1(G)}+\frac{s_2}{E_2(G)}\leq1.
\label{eq:hyperbolic}
\end{equation}
\end{conjecture}

It was shown in~\cite{barlev26coded} that this conjecture holds for:
\begin{itemize}
\item The identity code $G = I_k$.
\item All file-dedicated codes, where $N(F_1) + N(F_2) = n$ (this includes, in particular,  local MDS codes).
\item Global systematic MDS codes when $k \mid n$ (via dominance by local MDS).
\end{itemize}

In this work, we prove Conjecture~\ref{conj:hyperbolic} for every systematic code having at most $(k-2)/2$ columns that combine information from both files, that is, whenever
$$
N(F_1)+N(F_2)\geq n-\frac{k-2}{2}.
$$

\section{A Geometric Perspective on File Retrieval}\label{sec:geometric}

We introduce a geometric framework for analyzing the retrieval process by tracking the dimensions of the sampled subspace and its intersections with the two file subspaces. This viewpoint yields a pathwise relation between draws that advance both files and draws that advance neither.

\subsection{Dimension Tracking}
Let $V_0:=\{0\}$ and, for $t\geq 1$, let
$V_t:=\left\langle g_{\xi_1},\ldots,g_{\xi_t}\right\rangle$
denote the subspace spanned by the columns drawn up to time $t$. For $i\in\{1,2\}$, write $\tau_i:=\tau_{F_i}(G)$. We track the process through the three quantities
\[
(d_t,a_t,b_t)
:=
\bigl(\dim(V_t),\dim(V_t\cap F_1),\dim(V_t\cap F_2)\bigr).
\]

By the rank-nullity theorem applied to $\pi_1|_{V_t}$, whose kernel is $V_t\cap F_2$,
\[
d_t = \dim(\pi_1(V_t)) + \dim(V_t \cap F_2) = \dim(\pi_1(V_t)) + b_t.
\]
At the stopping time $\tau_1$, we have $F_1\subseteq V_{\tau_1}$ and hence $\pi_1(V_{\tau_1})=F_1$. Therefore,
\begin{equation}
d_{\tau_1}=s_1+b_{\tau_1},
\qquad
a_{\tau_1}=s_1.
\label{eq:state_at_tau1}
\end{equation}

\begin{lemma}\label{lem:increment}
Let $g:=g_{\xi_{t+1}}$ and suppose that $g\notin V_t$. Then
$d_{t+1}=d_t+1$. Moreover,
\begin{enumerate}
\item $a_{t+1}=a_t+1$ if and only if $\pi_2(g)\in\pi_2(V_t)$.
\item $b_{t+1}=b_t+1$ if and only if $\pi_1(g)\in\pi_1(V_t)$.
\end{enumerate}
\end{lemma}

\begin{IEEEproof}
The dimension $d_t$ increases by $1$ whenever $g \notin V_t$. 
Suppose first that $\pi_2(g)\in\pi_2(V_t)$. Then there exists $u\in V_t$ such that $\pi_2(u)=-\pi_2(g)$. Hence $u+g\in F_1$. Moreover, $u+g\notin V_t$, since otherwise $g\in V_t$. Therefore, $a_{t+1}=a_t+1$.

Conversely, if $a_{t+1}=a_t+1$, then there exist $u\in V_t$ and $\lambda\neq 0$ such that $u+\lambda g\in F_1$. Applying $\pi_2$ gives 
\[
\pi_2(u)+\lambda\pi_2(g)=0
\]
and hence
\[
\pi_2(g)=-\lambda^{-1}\pi_2(u)\in\pi_2(V_t).
\]
The second claim follows symmetrically.

\end{IEEEproof}

The lemma yields the following four cases.

\begin{corollary}
\label{cor:four_cases}
Let $g:=g_{\xi_{t+1}}$ and suppose that $g\notin V_t$. Then exactly one of the following occurs:

\begin{enumerate}
\item $\pi_2(g) \notin \pi_2(V_t)$ and $\pi_1(g) \notin \pi_1(V_t)$: only $d$ increments.
\item $\pi_2(g) \in \pi_2(V_t)$ and $\pi_1(g) \notin \pi_1(V_t)$: $d$ and $a$ increment.
\item $\pi_2(g) \notin \pi_2(V_t)$ and $\pi_1(g) \in \pi_1(V_t)$: $d$ and $b$ increment.
\item $\pi_2(g) \in \pi_2(V_t)$ and $\pi_1(g) \in \pi_1(V_t)$: $d$, $a$, and $b$ all increment.
\end{enumerate}
\end{corollary}

\begin{lemma}
\label{lem:conservation}
Let $\Gamma$ denote the number of draws up to time $\tau_1$ that fall into Case~4 of Corollary~\ref{cor:four_cases}. Then the number of draws that fall into Case~1 is also $\Gamma$.

\end{lemma}

\begin{IEEEproof}
A draw falls into one of Cases~1--4 exactly when it increases $d$.
Hence, by~\eqref{eq:state_at_tau1}, the total number of such draws up
to time $\tau_1$ is $d_{\tau_1}=s_1+b_{\tau_1}$.

Since $a_0=b_0=0$,  $\eqref{eq:state_at_tau1}$ implies that by time $\tau_1$ the quantities $a$ and $b$ have increased $s_1$ and $b_{\tau_1}$ times, respectively. Case~4 occurs $\Gamma$ times, and each such draw increases both $a$ and $b$. Hence, Case~2 occurs $s_1-\Gamma$ times and Case~3 occurs $b_{\tau_1}-\Gamma$ times.
It follows that the number
of Case~1 draws is
\[
d_{\tau_1}-(s_1-\Gamma)-(b_{\tau_1}-\Gamma)-\Gamma=\Gamma.
\]
\end{IEEEproof}

This balance at $\tau_1$ will be extended to a pathwise inequality in the next subsection. 

\subsection{The Defect and Its Consequences}

We extend Lemma~\ref{lem:conservation} to a pathwise inequality between Cases~1 and~4 that holds at every time $T$. We then derive a direct-sum decomposition at the recovery times, bound the number of simultaneous increments in terms of the number of mixed columns, and reformulate Conjecture~\ref{conj:hyperbolic} using quantities associated with the retrieval process.

The starting point is a single integer-valued invariant of the drawn subspace.

\begin{definition}
For $T\geq 0$, the \emph{defect} of $V_T$ is
\[
\begin{aligned}
\operatorname{def}(V_T)
&:=d_T-a_T-b_T\\
&=\dim(V_T)-\dim(V_T\cap F_1)-\dim(V_T\cap F_2).
\end{aligned}
\]
\end{definition}

\begin{lemma}\label{lem:defect_nonneg}
For every $T\ge 0$, $\operatorname{def}(V_T)\ge 0$.
\end{lemma}

\begin{IEEEproof}
Since $F_1\cap F_2=\{0\}$, the subspaces $V_T\cap F_1$ and $V_T\cap F_2$ also intersect trivially, so their sum is direct and has dimension $a_T+b_T$. This direct sum is contained in $V_T$, giving $a_T+b_T\le d_T$.
\end{IEEEproof}

\begin{lemma}\label{lem:defect_dyn}
A Case~1 draw increases the defect by $1$, a Case~4 draw decreases it by $1$, and every other draw leaves it unchanged. Consequently, if $C_i(T)$ denotes the number of Case~$i$ draws among the first $T$ draws, then
\begin{equation}\label{eq:defect_counts}
\operatorname{def}(V_T)=C_1(T)-C_4(T).
\end{equation}
\end{lemma}

\begin{IEEEproof}
If $g_{\xi_{t+1}}\in V_t$, then $V_{t+1}=V_t$, so the defect does not change. Otherwise, Corollary~\ref{cor:four_cases} shows that the changes in $(d,a,b)$ in Cases~1--4 are, respectively, $(1,0,0)$, $(1,1,0)$, $(1,0,1)$, and $(1,1,1)$. Thus, the corresponding changes in the defect are $1,0,0$, and $-1$. Since
$\operatorname{def}(V_0)=0$, summing these changes over the first $T$ draws gives~\eqref{eq:defect_counts}.
\end{IEEEproof}

\begin{theorem}[Pathwise Pairing Inequality]\label{thm:pairing}
For any rank-$k$ systematic generator matrix $G=[\,I_k\mid P\,]\in\mathbb{F}_q^{k\times n}$, any partition
$s_1+s_2=k$, and every $T\geq 0$,
\begin{equation}\label{eq:pathwise_pairing}
C_4(T)\leq C_1(T).
\end{equation}
\end{theorem}

\begin{IEEEproof}
By combining Lemmas~\ref{lem:defect_nonneg} and~\ref{lem:defect_dyn},
$$
C_1(T)-C_4(T)=\operatorname{def}(V_T)\ge 0.$$
\end{IEEEproof}

We next record several consequences of the preceding results, beginning with the behavior of the process at the recovery times.

\begin{corollary}\label{cor:def_zero_at_tau}
For each $i\in\{1,2\}$, at time $\tau_i$ we have
\[
\operatorname{def}(V_{\tau_i})=0
\qquad\text{and}\qquad
C_1(\tau_i)=C_4(\tau_i).
\]
\end{corollary}

The vanishing of the defect gives the following direct-sum decomposition at each recovery time.

\begin{corollary}\label{cor:purification}
For each $i\in\{1,2\}$,
\[
V_{\tau_i}
=
F_i\oplus\bigl(V_{\tau_i}\cap F_{3-i}\bigr).
\]
\end{corollary}

\begin{IEEEproof}
Since $F_i\subseteq V_{\tau_i}$, the direct sum on the right is contained in $V_{\tau_i}$.
By Corollary~\ref{cor:def_zero_at_tau} and the definition of the defect,
$
\dim(V_{\tau_i})
=
s_i+\dim(V_{\tau_i}\cap F_{3-i}).
$
Thus, both subspaces have the same dimension and are therefore equal.
\end{IEEEproof}

We next use the pathwise pairing inequality to bound the number of Case~4 draws in terms of the number of mixed columns.

\begin{corollary}\label{cor:gamma_bound}
Let
\[
M:=\bigl|\{j\in[n]:g_j\notin F_1\cup F_2\}\bigr|
\]
be the number of mixed columns of $G$. Then, for every $T\geq 0$,
\[
C_4(T)\leq C_1(T)\leq M.
\]
In particular, $\Gamma=C_4(\tau_1)\leq M$.
\end{corollary}

\begin{IEEEproof}
The first inequality follows from Theorem~\ref{thm:pairing}. If a draw falls into Case~1, then both projections of the drawn column are nonzero, so the column is mixed. Moreover, each column can contribute to Case~1 at most once, since after it is first drawn it belongs to the sampled subspace. Therefore, $C_1(T)\leq M$.
\end{IEEEproof}

The corollary connects the pathwise interaction between the two files to a fixed property of the code. A Case~4 draw increases both $a$ and $b$, but the total number of such draws is at most $M$. Thus, when $G$ has few mixed columns, simultaneous progress toward recovering the two files can occur only a limited number of times.

We next decompose each retrieval time into draws that increase the dimension of the sampled subspace and draws that do not. For $T\geq 0$, let $W(T)$ denote the number of draws among the first $T$ that do not increase $\dim(V_t)$.

\begin{corollary}\label{cor:tau_decomp}
For each $i\in\{1,2\}$,
\begin{equation}\label{eq:tau_decomp}
\tau_i
=
W(\tau_i)+s_i+\dim(V_{\tau_i}\cap F_{3-i}).
\end{equation}
Consequently,
\begin{equation}\label{eq:E_decomp}
E_i(G)
=
s_i+\mathbb E\!\left[\dim(V_{\tau_i}\cap F_{3-i})\right]
+\mathbb E[W(\tau_i)].
\end{equation}
The middle term in~\eqref{eq:E_decomp} equals
$\mathbb E[b_{\tau_1}]$ and $\mathbb E[a_{\tau_2}]$ for $i=1$ and $i=2$, respectively.
\end{corollary}

\begin{IEEEproof}
Among the first $T$ draws, exactly $W(T)$ do not increase the dimension of $V_t$, while each of the remaining $T-W(T)$ draws increases it by one. Since $d_0=0$, it follows that $T=W(T)+d_T.$ At time $\tau_i$, the rank-nullity identity gives
$
d_{\tau_i}
=
s_i+\dim(V_{\tau_i}\cap F_{3-i})$,
which proves~\eqref{eq:tau_decomp}. Taking expectations
gives~\eqref{eq:E_decomp}.
\end{IEEEproof}

Equation~\eqref{eq:E_decomp} separates $E_i(G)$ into three contributions: the $s_i$ dimensions required to recover $F_i$, the expected dimension of $V_{\tau_i}\cap F_{3-i}$ when $F_i$ is recovered, and the expected number of draws that do not increase the sampled subspace.

Since $b$ increases in Cases~3 and~4, while $a$ increases in
Cases~2 and~4, we have
\[
b_{\tau_1}=C_3(\tau_1)+C_4(\tau_1),
\qquad
a_{\tau_2}=C_2(\tau_2)+C_4(\tau_2).
\]
Therefore, Corollary~\ref{cor:gamma_bound} gives
\[
\mathbb E[b_{\tau_1}]
\leq \mathbb E[C_3(\tau_1)]+M,
\qquad
\mathbb E[a_{\tau_2}]
\leq \mathbb E[C_2(\tau_2)]+M.
\]

Finally, substituting~\eqref{eq:E_decomp} into Conjecture~\ref{conj:hyperbolic} shows that the conjecture is equivalent to
\begin{equation}\label{eq:hyperbolic_pathwise}
\frac{s_1}
{s_1+\mathbb E[b_{\tau_1}]+\mathbb E[W(\tau_1)]}
+
\frac{s_2}
{s_2+\mathbb E[a_{\tau_2}]+\mathbb E[W(\tau_2)]}
\leq 1.
\end{equation}

\section{Bounds in Terms of Mixed Columns}\label{sec:mixed_bounds}
The preceding section showed that the number of simultaneous increments is controlled by the number $M$ of mixed columns. We now obtain a complementary bound by projecting the sampled columns onto each file.
Spanning the corresponding projected space is necessary for recovering a file, and this gives lower bounds on $E_1(G)$ and $E_2(G)$. Combining these bounds will yield a sufficient condition for Conjecture~\ref{conj:hyperbolic} in terms of $M$.

Since $G$ has no zero columns, every column is either contained in $F_1$, contained in $F_2$, or mixed. Therefore,
\begin{equation}\label{eq:column_census}
N(F_1)+N(F_2)+M=n.
\end{equation}

For each $i\in\{1,2\}$, let
\begin{equation}\label{eq:mu_def}
\mu_i:=n-N(F_{3-i})
=
\bigl|\{j\in[n]:\pi_i(g_j)\neq 0\}\bigr|.
\end{equation}
Thus, $\mu_i$ is the number of columns with a nonzero component in $F_i$. Since $G$ is systematic, $\mu_i\ge s_i$. Moreover, each pure column contributes to exactly one of $\mu_1$ and $\mu_2$, while each mixed column contributes to both. Since $G$ has no zero columns, it follows that
\begin{equation}\label{eq:mu_sum}
\mu_1+\mu_2=n+M.
\end{equation}

For each $i\in\{1,2\}$, define
\[
\tau_i^{\mathrm{proj}}
:=
\min\{t\geq 0:\pi_i(V_t)=F_i\}.
\]
As shown in~\cite[Lem.~13]{barlev26coded},
$\tau_i^{\mathrm{proj}}\leq\tau_i$
and hence
\[
\mathbb E[\tau_i^{\mathrm{proj}}]\leq E_i(G).
\]

The bound in~\cite[Lem.~14]{barlev26coded} treats all $\mu_i$ columns with nonzero projection as potentially increasing the projected dimension at every stage. We sharpen this count by observing that once the projected span has dimension $\ell$, at least $\ell$ of these columns already belong to that span.

\begin{theorem}\label{thm:projected_bound}
For each $i\in\{1,2\}$,
\begin{equation}\label{eq:projected_bound}
E_i(G)
\geq
\mathbb E[\tau_i^{\mathrm{proj}}]
\geq
n\bigl(H_{\mu_i}-H_{\mu_i-s_i}\bigr).
\end{equation}
\end{theorem}

\begin{IEEEproof}
We prove the claim for $i=1$; the proof for $i=2$ is identical.
By~\cite[Lem.~13]{barlev26coded},
$E_1(G)\geq\mathbb E[\tau_1^{\mathrm{proj}}]$,
so it remains to bound the projected recovery time.

Suppose that the current projected span $\pi_1(V_t)$ has dimension $\ell<s_1$. Since $\pi_1(V_t)$ is spanned by the projections of the columns drawn so far, there are $\ell$ columns whose projections form a basis of $\pi_1(V_t)$. Thus, at least $\ell$ of the $\mu_1$ columns with nonzero projection already project into the current span.
Consequently, at most $\mu_1-\ell$ columns can increase its dimension.

Until the next dimension increase, the projected span remains unchanged. Hence, the probability that any given draw increases its dimension is at most $(\mu_1-\ell)/n$, and the expected waiting time for this increase is at least
$\frac{n}{\mu_1-\ell}$.

For $\ell\in\{0,\ldots,s_1-1\}$, let $Y_\ell$ denote the number of draws required for the projected span to increase from dimension $\ell$ to dimension $\ell+1$. Conditioned on the projected span when dimension $\ell$ is first reached, $Y_\ell$ is geometric with success probability at most $(\mu_1-\ell)/n$. Therefore,
\[
\mathbb E[Y_\ell]\geq\frac{n}{\mu_1-\ell}.
\]

Since the projected span must increase successively from dimension $0$ to dimension $s_1$,
\[
\tau_1^{\mathrm{proj}}
=
\sum_{\ell=0}^{s_1-1}Y_\ell.
\]
By linearity of expectation,
\begin{align*}
\mathbb E[\tau_1^{\mathrm{proj}}]
&=
\sum_{\ell=0}^{s_1-1}\mathbb E[Y_\ell]\\
&\geq
\sum_{\ell=0}^{s_1-1}\frac{n}{\mu_1-\ell}\\
&=
n\bigl(H_{\mu_1}-H_{\mu_1-s_1}\bigr).
\end{align*}
\end{IEEEproof}

We now combine the bounds for the two files using \eqref{eq:mu_sum}.

\begin{corollary}\label{cor:mixed_bound}
For every systematic generator matrix $G$,
\begin{equation}\label{eq:mixed_bound}
\frac{s_1}{E_1(G)}+\frac{s_2}{E_2(G)}
\leq
1+\frac{M-(k-2)/2}{n}.
\end{equation}
\end{corollary}

\begin{IEEEproof}
Fix $i\in\{1,2\}$. By Theorem~\ref{thm:projected_bound},
\[
\frac{s_i}{E_i(G)}
\leq
\frac{1}{n}
\left(
\frac{s_i}{H_{\mu_i}-H_{\mu_i-s_i}}
\right).
\]
Since $\mu_i\geq s_i$, the $s_i$ numbers
\[
\mu_i,\mu_i-1,\ldots,\mu_i-s_i+1
\]
are positive. By the definition of the harmonic numbers, the expression
in parentheses is their harmonic mean. Their arithmetic mean is
\begin{align*}
\frac{1}{s_i}\sum_{\ell=0}^{s_i-1}(\mu_i-\ell)
=
\mu_i-\frac{1}{s_i}\sum_{\ell=0}^{s_i-1}\ell
=
\mu_i-\frac{s_i-1}{2}.
\end{align*}
Since the harmonic mean is at most the arithmetic mean,
\[
\frac{s_i}{H_{\mu_i}-H_{\mu_i-s_i}}
\leq
\mu_i-\frac{s_i-1}{2}.
\]
It follows that
\[
\frac{s_i}{E_i(G)}
\leq
\frac{1}{n}
\left(\mu_i-\frac{s_i-1}{2}\right).
\]

Applying this inequality for $i=1,2$, and using
$\mu_1+\mu_2=n+M$ and $s_1+s_2=k$, gives
\begin{align*}
\frac{s_1}{E_1(G)}+\frac{s_2}{E_2(G)}
&\leq
\frac{\mu_1+\mu_2-(k-2)/2}{n}\\
&=
1+\frac{M-(k-2)/2}{n}.
\end{align*}\vspace{-1.5ex}
\end{IEEEproof}

Corollary~\ref{cor:mixed_bound} immediately gives our main result.

\begin{theorem}\label{thm:main}
Conjecture~\ref{conj:hyperbolic} holds for every rank-$k$ systematic
generator matrix with no zero columns whenever the number~$M$ of mixed
columns satisfies
\[
M\leq\frac{k-2}{2}.
\]
For equal file dimensions $s_1=s_2=s$, this condition becomes
$M\leq s-1$. In particular, the conjecture holds when $M\leq1$ and
$k\geq4$.
\end{theorem}

When $k=3$, the assumption $\max\{s_1,s_2\}\geq2$ leaves only the partition $\{s_1,s_2\}=\{1,2\}$. For  $n=k+1=4$, the remaining case $M=1$ can also be proved by a more detailed direct analysis; due to space limitations, we defer the proof to the long version of this work.  

\section{Conclusion}\label{sec:conclusion}

We studied the expected retrieval times of two files of dimensions $s_1$ and $s_2$ encoded by a systematic $[n,k]$ linear code over an arbitrary finite field, where $k=s_1+s_2$, and assumed that the generator matrix has no zero columns. By tracking the sampled span and its intersections with the two file subspaces, we obtained a pathwise comparison between samples that advance both files and samples that advance neither. We also sharpened a projection bound for the expected retrieval times. If $M$ denotes the number of mixed columns, this gives
\vspace{-1ex}
$$
\frac{s_1}{E_1(G)}+\frac{s_2}{E_2(G)}
\leq
1+\frac{M-(k-2)/2}{n}.
$$
Consequently, whenever
$M\leq\frac{k-2}{2}$,
we obtain the conjectured hyperbolic bound
$$
\frac{s_1}{E_1(G)}+\frac{s_2}{E_2(G)}\leq1.
$$
For equal file dimensions $s_1=s_2=s$, the sufficient condition becomes
$M\leq s-1$.

In particular the result covers $M=1$ whenever $k\geq4$. The case $k=3$, $n=4$, and $M=1$ can be handled by a more detailed direct analysis. Our current efforts focus on the remaining cases with small values of $M$ and $k$ and, more generally, on strengthening the bound outside the range $M\leq(k-2)/2$. The projection argument used here records only the number of columns having a nonzero component in each file subspace. We are examining whether the supports and linear dependencies of the mixed columns, together with the pathwise comparison, can provide stronger bounds.\\
\noindent Other directions include:
\begin{itemize}
\item Determine whether file-dedicated MDS codes are Pareto-optimal at finite lengths. Equivalently, can a code with the same file dimensions and code length improve the expected retrieval time of one file relative to a local MDS code without increasing that of the other?
\item Extend the results to nonsystematic codes and to more than two files. For the case of several files, this first requires identifying the appropriate analogue of the hyperbolic bound.
\item Study nonuniform sampling distributions, which model differences in molecule abundance caused by synthesis and amplification biases. In this setting, the retrieval times depend on both the generator matrix and the sampling probabilities, and a weighted analogue of the hyperbolic bound may be needed.
\item Extend the model to noisy retrieval, where sampled strands may be read incorrectly. The stopping condition must then account for both the sampled linear information and the error-correction capability of the code.
\end{itemize}

\newpage
\IEEEtriggeratref{12} 
\bibliographystyle{IEEEtran}
\bibliography{refs.bib}

\begin{thebibliography}{10}
\providecommand{\url}[1]{#1}
\csname url@samestyle\endcsname
\providecommand{\newblock}{\relax}
\providecommand{\bibinfo}[2]{#2}
\providecommand{\BIBentrySTDinterwordspacing}{\spaceskip=0pt\relax}
\providecommand{\BIBentryALTinterwordstretchfactor}{4}
\providecommand{\BIBentryALTinterwordspacing}{\spaceskip=\fontdimen2\font plus
\BIBentryALTinterwordstretchfactor\fontdimen3\font minus \fontdimen4\font\relax}
\providecommand{\BIBforeignlanguage}[2]{{%
\expandafter\ifx\csname l@#1\endcsname\relax
\typeout{** WARNING: IEEEtran.bst: No hyphenation pattern has been}%
\typeout{** loaded for the language `#1'. Using the pattern for}%
\typeout{** the default language instead.}%
\else
\language=\csname l@#1\endcsname
\fi
#2}}
\providecommand{\BIBdecl}{\relax}
\BIBdecl

\bibitem{church2012next}
G.~M. Church, Y.~Gao, and S.~Kosuri, ``Next-generation digital information storage in {DNA},'' \emph{Science}, vol. 337, no. 6102, pp. 1628--1628, 2012.

\bibitem{goldman2013towards}
N.~Goldman, P.~Bertone, S.~Chen, C.~Dessimoz, E.~M. LeProust, B.~Sipos, and E.~Birney, ``Towards practical, high-capacity, low-maintenance information storage in synthesized {DNA},'' \emph{nature}, vol. 494, no. 7435, pp. 77--80, 2013.

\bibitem{grass2015robust}
R.~N. Grass, R.~Heckel, M.~Puddu, D.~Paunescu, and W.~J. Stark, ``Robust chemical preservation of digital information on {DNA} in silica with error-correcting codes,'' \emph{Angewandte Chemie International Edition}, vol.~54, no.~8, pp. 2552--2555, 2015.

\bibitem{erlich2017dna}
Y.~Erlich and D.~Zielinski, ``{DNA} fountain enables a robust and efficient storage architecture,'' \emph{science}, vol. 355, no. 6328, pp. 950--954, 2017.

\bibitem{organick2018random}
L.~Organick, S.~D. Ang, Y.-J. Chen, R.~Lopez, S.~Yekhanin, K.~Makarychev, M.~Z. Racz, G.~Kamath, P.~Gopalan, B.~Nguyen \emph{et~al.}, ``Random access in large-scale {DNA} data storage,'' \emph{Nature Biotechnology}, vol.~36, no.~3, pp. 242--248, 2018.

\bibitem{bar2025scalable}
D.~Bar-Lev, I.~Orr, O.~Sabary, T.~Etzion, and E.~Yaakobi, ``Scalable and robust {DNA}-based storage via coding theory and deep learning,'' \emph{Nature Machine Intelligence}, vol.~7, no.~4, pp. 639--649, 2025.

\bibitem{milenkovic2024dna}
O.~Milenkovic and C.~Pan, ``{DNA}-based data storage systems: A review of implementations and code constructions,'' \emph{IEEE Transactions on Communications}, vol.~72, no.~7, pp. 3803--3828, 2024.

\bibitem{sabary2024survey}
O.~Sabary, H.~M. Kiah, P.~H. Siegel, and E.~Yaakobi, ``Survey for a decade of coding for {DNA} storage,'' \emph{IEEE Transactions on Molecular, Biological, and Multi-Scale Communications}, vol.~10, no.~2, pp. 253--271, 2024.

\bibitem{heckel2019characterization}
R.~Heckel, G.~Mikutis, and R.~N. Grass, ``A characterization of the {DNA} data storage channel,'' \emph{Scientific reports}, vol.~9, no.~1, p. 9663, 2019.

\bibitem{erdHos1961classical}
P.~Erd{\H{o}}s and A.~R{\'e}nyi, ``On a classical problem of probability theory,'' \emph{A Magyar Tudom{\'a}nyos Akad{\'e}mia Matematikai Kutat{\'o} Int{\'e}zet{\'e}nek K{\"o}zlem{\'e}nyei}, vol.~6, no. 1-2, pp. 215--220, 1961.

\bibitem{flajolet1992birthday}
P.~Flajolet, D.~Gardy, and L.~Thimonier, ``Birthday paradox, coupon collectors, caching algorithms and self-organizing search,'' \emph{Discrete Applied Mathematics}, vol.~39, no.~3, pp. 207--229, 1992.

\bibitem{barlev2025cover}
D.~Bar-Lev, O.~Sabary, R.~Gabrys, and E.~Yaakobi, ``Cover your bases: How to minimize the sequencing coverage in {DNA} storage systems,'' \emph{IEEE Transactions on Information Theory}, vol.~71, no.~1, pp. 192--218, 2025.

\bibitem{hanna2026reliability}
S.~K. Hanna, ``On the reliability of information retrieval from {MDS} coded data in {DNA} storage,'' \emph{IEEE Transactions on Molecular, Biological, and Multi-Scale Communications}, 2026.

\bibitem{cao2026modeling}
R.~Cao, P.~Zhou, and X.~Chen, ``Modeling dna storage retrieval reliability via sequencing coverage depth,'' \emph{Briefings in Bioinformatics}, vol.~27, no.~3, p. bbag303, 2026.

\bibitem{bertuzzo2026duality}
M.~Bertuzzo, A.~Ravagnani, and E.~Yaakobi, ``The {DNA} coverage depth problem: Duality, weight distributions, and applications,'' \emph{arXiv preprint arXiv:2603.06489}, 2026.

\bibitem{grunbaum2026general}
Y.~Grunbaum and E.~Yaakobi, ``General coverage models-structure, monotonicity and shotgun sequencing,'' in \emph{2026 IEEE International Symposium on Information Theory (ISIT)}.\hskip 1em plus 0.5em minus 0.4em\relax IEEE, 2026, pp. 1--6.

\bibitem{gruica2025combinatorial}
A.~Gruica, D.~Bar-Lev, A.~Ravagnani, and E.~Yaakobi, ``A combinatorial perspective on random access efficiency for {DNA} storage,'' \emph{IEEE Transactions on Information Theory}, 2025.

\bibitem{gruica2026geometry}
A.~Gruica, M.~Montanucci, and F.~Zullo, ``The geometry of codes for random access in {DNA} storage,'' \emph{Designs, Codes and Cryptography}, vol.~94, no.~5, p. 114, 2026.

\bibitem{boruchovsky2026making}
A.~Boruchovsky, O.~Elishco, R.~Gabrys, A.~Gruica, I.~Tamo, and E.~Yaakobi, ``Making it to first: The random access problem in {DNA} storage,'' \emph{IEEE Transactions on Information Theory}, 2026.

\bibitem{bodur2026random}
{\c{S}}.~Bodur, S.~Lia, H.~H. L{\'o}pez, R.~Ludhani, A.~Ravagnani, and L.~Seccia, ``The random variables of the {DNA} coverage depth problem,'' \emph{IEEE Transactions on Information Theory}, 2026.

\bibitem{wang2026random}
C.~Wang and E.~Yaakobi, ``Random access in {DNA} storage: Algorithms, constructions, and bounds,'' \emph{arXiv preprint arXiv:2601.07053}, 2026.

\bibitem{abraham2024covering}
H.~Abraham, R.~Gabrys, and E.~Yaakobi, ``Covering all bases: The next inning in {DNA} sequencing efficiency,'' in \emph{2024 IEEE International Symposium on Information Theory (ISIT)}.\hskip 1em plus 0.5em minus 0.4em\relax IEEE, 2024, pp. 464--469.

\bibitem{gruica2026generalized}
A.~Gruica, A.~Petrillo, and F.~Zullo, ``The generalized random access problem for linear codes,'' \emph{arXiv preprint arXiv:2608.20152}, 2026.

\bibitem{barlev26coded}
D.~Bar-Lev, ``Coded information retrieval for block-structured {DNA}-based data storage,'' \emph{arXiv preprint arXiv:2603.17154}, 2026.

\end{thebibliography}

\end{document}